\documentclass[a4paper,USenglish]{article}
\usepackage{amsmath,amssymb,amsthm,mathtools,bm}
\usepackage{thm-restate}
\usepackage{hyperref}
\usepackage{cleveref}
\usepackage{booktabs,multirow}
\usepackage[dvipsnames]{xcolor}
\usepackage[linesnumbered,ruled,vlined]{algorithm2e}

\SetCommentSty{mycommfont}
\newtheorem{theorem}{Theorem}

\newtheorem{lemma}{Lemma}
\newtheorem{corollary}{Corollary}
\newtheorem{remark}{Remark}

\newcommand{\OPT}{\mathrm{OPT}}
\newcommand{\ALG}{\mathrm{ALG}}
\newcommand{\IGNORE}{\mathrm{IGNORE}}
\newcommand{\pmax}{p_{\max}}
\newcommand{\argmax}{\mathop{\rm arg\,max}}

\newcommand{\ot}{\leftarrow}
\renewcommand{\mid}{:}
\renewcommand{\O}{\mathrm{O}}
\newcommand{\shortslash}{\scalebox{.8}{/}}
\usepackage{comment}
\usepackage{cite}
\usepackage{fullpage}
\usepackage[sort&compress,square,numbers]{natbib}

\crefname{algocf}{alg.}{algs.}
\Crefname{algocf}{Algorithm}{Algorithms}

\title{Scheduling on Identical Machines with Setup Time and Unknown Execution Time}

\usepackage{authblk}
\author[1]{Yasushi Kawase}
\affil{The University of Tokyo, Japan}
\author[2]{Kazuhisa Makino}
\affil{Kyoto University, Japan}
\author[3]{Vinh Long Phan}
\affil{Toyota Motor Corporation, Japan}
\author[4]{Hanna Sumita}
\affil{Institute of Science Tokyo, Japan}
\date{}

\begin{document}

\maketitle

\begin{abstract}
In this study, we investigate a scheduling problem on identical machines in which jobs require initial setup before execution. We assume that an algorithm can dynamically form a batch (i.e., a collection of jobs to be processed together) from the remaining jobs. The setup time is modeled as a known monotone function of the set of jobs within a batch, while the execution time of each job remains unknown until completion. This uncertainty poses significant challenges for minimizing the makespan. We address these challenges by considering two scenarios: each job batch must be assigned to a single machine, or a batch may be distributed across multiple machines. For both scenarios, we analyze settings with and without preemption. Across these four settings, we design online algorithms that achieve asymptotically optimal competitive ratios with respect to both the number of jobs and the number of machines.
\end{abstract}

\section{Introduction}
Efficient job allocation across multiple workers or machines is crucial for optimizing productivity in industrial environments. 
For example, consider distributing computational jobs across multiple virtual machines.
The processing time for a batch of jobs consists of two components: the \emph{setup time}, which includes configuring the environment or installing necessary software or libraries, and the \emph{execution time}, which represents the duration of performing the actual tasks.
While setup times are typically known in advance, the execution time of each job is often unpredictable until the job is completed.
This uncertainty complicates the design of scheduling algorithms aimed at minimizing the \emph{makespan}, which is the time when all jobs are completed.

In this paper, we introduce a scheduling problem involving $n$ jobs on $m$ identical machines with known setup times and unknown execution times, referred to as the \emph{unknown execution time scheduling (UETS)} problem.
We assume that an algorithm can dynamically form a \emph{batch}---a collection of jobs that are processed together---from the remaining jobs.
The setup time for a batch is defined as a set function over subsets of jobs, representing the total time required to prepare the jobs in the batch for processing.
We explore two primary scenarios: one called \emph{sUETS}, in which each constructed batch must be assigned to a single machine, and the other called \emph{mUETS}, in which a batch can be distributed across multiple machines.
Additionally, we examine two settings based on whether preemption is allowed.
In the non-preemptive setting, once a batch process is started, it must be run without interruption until completion.
In the preemptive setting, the processing of a batch can be interrupted.
After an interruption, only uncompleted jobs are regrouped into batches and the process is resumed from the setup phase (completed jobs do not need to be processed again).
Note that, the non-preemptive setting generally requires longer processing times compared to the preemptive one, because preemption allows for more flexible handling of incomplete batches and may reduce redundant processing.

Our objective is to design online algorithms for these scheduling problems. We analyze the performance of an online algorithm by the competitive ratio, which is the worst-case ratio between the makespan achieved by the online algorithm and that of an optimal offline algorithm.
We assume that offline algorithms have complete knowledge of each job's execution time in advance.
We refer to the schedule produced by the optimal offline algorithm as the \emph{optimal schedule} and its makespan as the \emph{optimal makespan}.
We refer to the makespan of the schedule produced by an algorithm simply as the algorithm's makespan.
An online algorithm is said to be \emph{$\rho$-competitive} if its makespan is at most $\rho$ times the optimal makespan for any instance.
We will design online algorithms with asymptotically optimal competitive ratios with respect to both the number of machines and the number of jobs.
%To focus on the hardness introduced by uncertainty, we basically ignore issues of computational complexity in this study.

\subsection{Our results}
We study the competitive ratios for the scheduling problem with setup time and unknown execution times. 
We first show that any algorithm designed for jobs with a release time of $0$ can be adapted to handle arbitrary release times at the expense of only a constant factor increase in the competitive ratio.
We formally state this in \Cref{thm:ignore}.
Thus, we may assume that the release time of each job is $0$, i.e., jobs can start processing immediately at time $0$. 
A summary of our results is provided in \Cref{tab:results}.

For the non-preemptive sUETS problem, we present two algorithms. First, we construct an algorithm with a competitive ratio of $m$ (\Cref{thm:single-np-m}). Second, we design an algorithm that is $\O\big(\sqrt{n\shortslash m}\big)$-competitive (\Cref{thm:single-np-nm}). 
By combining these two algorithms, we obtain an $\O(n^{1/3})$-competitive algorithm (\Cref{cor:single-np-n}).
Moreover, we prove lower bounds showing that every online algorithm for this problem must have a competitive ratio of at least $\Omega(m)$ (\Cref{thm:single-p-hard}) and $\Omega(n^{1/3})$ (\Cref{thm:single-np-hard}).
For the preemptive sUETS problem,  we design an algorithm with a competitive ratio of $\O(\log n/\log\log n)$ (\Cref{thm:single-p-n}), and we demonstrate that the competitive ratios of $\O(m)$ and $\O(\log n/\log\log n)$ are best possible (\Cref{thm:single-p-hard}). 

Turning to the non-preemptive mUETS problem, we first construct an $\O(\sqrt{m})$-competitive algorithm (\Cref{thm:multiple-np}).
By integrating this with the $\O\big(\sqrt{n\shortslash m}\big)$-competitive algorithm for the non-preemptive sUETS problem (\Cref{thm:single-np-nm}), we derive an $\O(n^{1/4})$-competitive algorithm for the non-preemptive mUETS problem (\Cref{cor:multiple-np-n}). We further prove that $\O(\sqrt{m})$ and $\O(n^{1/4})$ are optimal (\Cref{thm:multiple-np-hard}).
Finally, for the preemptive mUETS problem, we establish that the best possible competitive ratios are
$\Theta(\log m/\log\log m)$ and $\Theta(\log n/\log\log n)$ (\Cref{thm:single-p-n,thm:multiple-p,thm:multiple-p-hard}).

%Due to space limitations, some proofs are deferred to \Cref{sec:omitted}.

\begin{table}[t]
    \centering
    \caption{The competitive ratios of the UETS problems with $n$ jobs and $m$ machines.}
    \label{tab:results}
    \renewcommand{\arraystretch}{1.2}
    \tabcolsep = 4mm
    \begin{tabular}{c|cc}
    \toprule
               & single                                             & multiple \\
    \midrule
    \multirow{2}{*}{non-preemptive}& $\Theta(m)$ (Thms.~\ref{thm:single-np-m} and \ref{thm:single-np-hard})& $\Theta(\sqrt{m})$ (Thms.~\ref{thm:multiple-np} and \ref{thm:multiple-np-hard})\\[2pt]

    & $\Theta(n^{1/3})$ (Cor.~\ref{cor:single-np-n} and Thm.~\ref{thm:single-np-hard})& $\Theta(n^{1/4})$ (Cor.~\ref{cor:multiple-np-n} and Thm.~\ref{thm:multiple-np-hard})\\[2pt]\hline
    \multirow{2}{*}{preemptive}    & $\Theta(m)$ (Thms.~\ref{thm:single-np-m} and \ref{thm:single-p-hard})& $\Theta\big(\frac{\log m}{\log\log m}\big)$ (Thms.~\ref{thm:multiple-p} and \ref{thm:multiple-p-hard})\\[2pt]
    & $\Theta\big(\frac{\log n}{\log\log n}\big)$ (Thms.~\ref{thm:single-p-n} and \ref{thm:single-p-hard})&$\Theta\big(\frac{\log n}
{\log\log n}\big)$ (Thms.~\ref{thm:single-p-n} and \ref{thm:multiple-p-hard})\\
    \bottomrule
    \end{tabular}
\end{table}

\subsection{Related work}\label{subsec:relatedwork}

% no set up time
Classical scheduling problems typically assume that the setup time is $0$ for every job batch.
In this case, the well-known \emph{list scheduling algorithm} achieves $(2-1/m)$-competitive, which is proven to be optimal~\cite{Graham1966,Gusfield1984,HS1989,SWW1995}.
This algorithm assigns unprocessed jobs to any available machine in the order they appear on the job list, disregarding execution times.

When execution times are unknown until job completion, the problem falls under the category of \emph{non-clairvoyant scheduling}~\cite{motwani1994nonclairvoyant}. Recent studies have explored non-clairvoyant scheduling in input prediction models~\cite{bampis2023non,im2023non}. Shmoys et al.~\cite{SWW1995} introduced a technique to transform a $\rho$-competitive algorithm for a scheduling problem without release times into a $2\rho$-competitive algorithm for the same problem with release times. For comprehensive overviews of online scheduling, see surveys~\cite{Sgall1998,pruhs2003online} and the book by Pinedo~\cite{Pinedo2008}.

% setup time
Scheduling problems that incorporate both setup times and execution times arise in various applications, such as cloud computing (where virtual machines must be initialized based on job types) and and production systems (where machines require reconfiguration, such as changing molds or colors)~\cite{sahney1972single,MCMQ1999,HNQ2016,GN2000,divakaran2000online,allahverdi2008survey}.
For example, in plastic production systems, attaching a specific mold to a machine constitutes the setup time. If consecutive jobs use the same mold, no additional setup time is required.
The total setup time for a batch can be precomputed as the minimum time needed for attaching, exchanging, and removing molds.
% Generally, this setup time can be represented as follows: (i) the jobs are partitioned into distinct types, (ii) each type is associated with a specific setup time, and (iii) the setup time for a batch is calculated as the sum of the setup times for the types included in the batch. 
% We refer to the setup time defined in this way as the \emph{type-specific setup time}.
% Our model is capable of accommodating this type of setup time.

Gambosi and Nicosia~\cite{GN2000} studied an online version of scheduling with setup times in the one-by-one model.
M\"{a}cker et al.~\cite{macker2017non} investigated non-clairvoyant scheduling with setup times and proposed an $\O(\sqrt{n})$-competitive algorithm for minimizing maximum flow time on a single machine.
Dogeas et al.~\cite{Dogeas+2024} considered a scenario where the execution time of each job is only revealed after an obligatory test with a known duration. This test time can be viewed as a kind of setup time, although it differs in that the execution time becomes apparent.

% Ours
Goko et al.~\cite{goko2022online} introduced 
a scheduling model with a metric state space that involves setup time and unknown execution time settings together. 
Their model includes a scheduling problem faced by repair companies where each worker needs to visit customers' houses, do the repair jobs, and then return to the office. The setup time corresponds to the shortest tour length for the customers' houses, and the execution time corresponds to the duration of the repair jobs.
%Their model also includes the online traveling salesman problem (TSP)~\cite{AF+1995,AF+2001,BK+2001,JW2008,BH+2021} and the online dial-a-ride problem~\cite{AKR2000,Krumke2001,FS2001,LLPSS2004,BLS2006,BDS2019,Birx2020,BD2020}.
Their model also includes the online dial-a-ride problem~\cite{AKR2000,Krumke2001,FS2001,LLPSS2004,BLS2006,BDS2019,Birx2020,BD2020}, in which taxis are offered to pick up and drop passengers for transportation jobs.
Our preemptive mUETS problem can be seen as an abstraction of their model, and hence, a similar approach can be used to solve it. However, the other problems (i.e., preemptive/non-preemptive sUETS and non-preemptive mUETS) are different and require other approaches.

% offline
The setting in which all jobs have an execution time of $0$ and are released at time $0$, with only the setup time considered as processing time, can be viewed as an offline load balancing problem. 
This problem is NP-hard even if the setup time is additive and the number of machines is two, as this special case corresponds to the PARTITION problem~\cite{garey1979computers}.
For the additive setup time, the list scheduling algorithm works as a $(2-1/m)$-approximation algorithm. Moreover, this problem admits
a polynomial time approximation scheme (PTAS)~\cite{HS1987}.
Svitkina and Fleischer~\cite{SvitkinaFleischer2011} presented an $\O(\sqrt{n/\log n})$-approximation algorithm for the submodular setup times, and demonstrated that this is best possible.
Furthermore, Nagano and Kishimoto~\cite{nagano2019subadditiveloadbalancing} provided a $2\cdot (\max_{j\in[m]}\frac{|S_j^*|}{1+(|S_j^*|-1)(1-\kappa_c(S_j^*))})$-approximation algorithm for the subadditive setup times, where $(S_1^*,\dots,S_m^*)$ is an optimal partition and $\kappa_c(S)$ is the curvature of $c$ at $S\subseteq J$. %They also provided some impossibility.

\section{Preliminaries}
For a positive integer $k$, we write $[k]$ to denote the set $\{1,2,\dots,k\}$.
We are given a set $J$ of $n$ jobs and $m$ identical machines.
We assume that all jobs are given and released at time $0$.
As we will show in \Cref{subsec:release}, this assumption only increases a constant factor in the competitive ratios.
We denote the set of machines by $[m]$.
We assume that $n\ge m\ge 2$, as the optimal scheduling is clear otherwise.
The jobs are executed in batches, and each batch's processing time consists of two components: the \emph{setup time} and the \emph{execution time}.

The setup time is the total time one machine takes to set up jobs in the batch.
For a batch $X\subseteq J$, the setup time for $X$ is represented as $c(X) \in \mathbb{R}_+$.
The execution time for each job $j\in J$, denoted by $p_j$, is unknown until its process is completed.
We write $p(X)$ to denote the total execution time of a batch $X\subseteq J$, i.e., $p(X)=\sum_{j\in X}p_j$.
Consequently, the overall processing time of a batch $X$ on a single machine is given by $p(X)+c(X)$.
We assume that the setup time $c(X)$ for every $X\subseteq J$ is available in advance, while the execution time of each job is unknown until its completion.

We also deal with a situation where one batch can be assigned to multiple machines.
We assume that when a batch $X$ is assigned to $k$ machines, each of the machines incurs a setup time $c(X)$, and the jobs are executed on a first-come, first-served basis.
For example, when a batch $X$ is processed in a production system, we assume that each machine forms a queue of the attachments needed for jobs in $X$, and proceeds as follows: (i) dequeues and installs the first attachment in the queue, (ii) repeatedly executes an unprocessed job that is compatible with the current attachment as long as such a job exists, and (iii) removes the current attachment and returns to step (i) if the attachment queue is not empty.
The minimum processing time for batch $X$ across $k$ machines falls within the range of $c(X)+p(X)/k$ and $c(X)+p(X)/k+\max_{j\in X}p_j$.
% One may be aware that setup times can be shortened by omitting unnecessary reassignment of attachments.
% In \Cref{sec:discuss}, we will observe that our algorithms remain asymptotically optimal even with such omission.

At each time, we can assign a batch of jobs that have not yet been completed and are not assigned anywhere else.
In the sUETS and mUETS problems, a batch can be allocated to a single machine and multiple machines, respectively.
Our objective is to minimize the makespan, which is the time at which all the jobs are completed.

We assume that the setup time function $c\colon 2^J\to\mathbb{R}_+$ is \emph{monotone subadditive}. 
Monotonicity means that $c(X)\le c(Y)$ for any $X\subseteq Y\subseteq J$.
Subadditivity means that for any disjoint subsets $X$ and $Y$ from $J$, the sum of their setup times is greater than or equal to the setup time of their union, i.e., $c(X)+c(Y)\ge c(X\cup Y)$.
Note that this assumption of subadditivity does not lose generality.
To observe this, for any (not necessarily subadditive) function $c$, define the function $\bar{c}(X)=\min_{(X_1,\dots,X_\ell):\,\text{partition of $X$}}\sum_{i=1}^\ell c(X_i)$.
This new function $\bar{c}$ is monotone and subadditive by construction---it represents the minimum total setup time achievable if one were allowed to split the batch $X$ into sub-batches and incur the setup cost separately for each. Importantly, replacing $c$ with $\bar{c}$ in our analysis does not change the fundamental nature of the scheduling problem, and thus we assume without loss of generality that the setup time function is monotone and subadditive.

% Here, we consider a partition of batch $X$ into sub-batches $X_1,\dots,X_\ell$, and calculate the total setup time by summing up the individual setup times of each sub-batch.
% By doing so, we employ the minimum total time among all possible partitions.

We explore two settings based on whether an allocated batch can be preempted.
In the non-preemptive setting, once a batch is assigned, it must be processed to completion without interruption.
Conversely, in the preemptive setting, the processing of a batch can be halted and canceled. 
Suppose that a batch $X$ is assigned to $k$ machines and is preempted after $t$ units of time from the start of its processing.
In this case, we assume that the machines will become available to process another batch after at most $c(X)+\max_{j\in X}p_j$ units of time.
This assumption ensures that the processing time does not increase even if preemption occurs during the setup phase.
Furthermore, at the moment of preemption, let $X'\subseteq X$ denote the set of jobs that have been completed. Then, the following inequality holds: $c(X)+p(X')/k+\max_{j\in X}p_j\ge t$.
This condition ensures that the total time spent up to the point of preemption is consistent with the setup and execution times for the completed jobs.

Note that, in any setting, we can assume that an optimal schedule does not perform preemptions and assigns only one batch to each machine at the beginning because the setup time $c$ is subadditive.
Thus, an optimal schedule can be represented as a partition of jobs $(X_1^*,\dots,X_m^*)$ and the optimal makespan is $\max_{i\in[m]} (c(X_i^*)+p(X^*_i))$ independently of the settings.
Based on this observation, we derive a lower bound for the optimal makespan, as formalized in the following lemma. 
\begin{lemma}\label{lem:opt-lower}
The optimal makespan is at least 
\begin{align*}
\max\left\{\min_{(X_1,\dots,X_m):\,\text{partition of $J$}}\ \max_{i\in[m]}\ c(X_i),\ (c(J)+p(J))/m,\ \max_{j\in J}(c(\{j\})+p_j)\right\}.
\end{align*}
\end{lemma}
\begin{proof}
Let $(X_1^*,\dots,X_m^*)$ be an optimal schedule.
Then, the optimal makespan is at least
\begin{align*}
\max_{i\in[m]} \left(c(X_i^*)+p(X^*_i)\right)
\ge \max_{i\in[m]} c(X_i^*)
\ge \min_{(X_1,\dots,X_m):\,\text{partition of $J$}}\ \max_{i\in[m]}\ c(X_i).
\end{align*}
Additionally, we have
\begin{align*}
\max_{i\in[m]} \left(c(X_i^*)+p(X^*_i)\right)
\ge \max_{j\in J}\left(c(\{j\})+p_j\right)
\ge (c(J)+p(J))/m
\end{align*}
by monotone subadditivity of the setup time.
Therefore, we have the lower bound.
\end{proof}

\subsection{Examples of setup times}\label{sec:setup times}
This subsection illustrates several examples of setup times raised in practical applications.
In addition, we discuss the approximability of the partition (load balancing) problem
\begin{align}
\min_{(X_1,\dots,X_m):\,\text{partition of $J$}}\  \max_{i\in[m]}\ c(X_i) \label{prob:partition}
\end{align}
for each class of setup times.

\medskip
\noindent\textbf{Constant setup times}\quad 
One of the simplest examples of the setup time family is the constant setup time, where the setup time is a certain constant independent of the batch. 
Formally, the setup time is represented as $c(X)=1$ for $X\in 2^J\setminus\{\emptyset\}$ and $c(\emptyset)=0$.
It is not difficult to see that any partition of the jobs $(X_1,\dots,X_m)$ is optimal for the problem \eqref{prob:partition}.

\medskip
\noindent\textbf{Type-specific setup times}\quad 
In production systems with attachments, jobs are partitioned based on their types. 
Let $T$ be the set of types and $(J_t)_{t\in T}$ is the partition of jobs with types $T$, i.e., $\bigcup_{t\in T}J_t=J$ and $J_t\cap J_{t'}=\emptyset$ for all distinct $t,t'\in T$. 
Let $w_t$ be the setup time for type $t$.
Then, the setup time for $X\subseteq J$ is defined as $c(X)=\sum_{t\in T:\,X\cap J_t\ne\emptyset}w_t$.
Additionally, we say that a type-specific setup time $c$ is unweighted if $w_t=1$ for all $t\in T$, i.e., $c(X)=|\{t\in T\mid X\cap J_t\ne\emptyset\}|$ for $X \subseteq J$.
% The unweighted setup time for $X\subseteq J$ is $c(X)=|\{t\in T\mid X\cap J_t\ne\emptyset\}|$.
% We call this unweighted type-specific setup times.
% 
% Generally, this setup time can be represented as follows: (i) the jobs are partitioned into distinct types, (ii) each type is associated with a specific setup time, and (iii) the setup time for a batch is calculated as the sum of the setup times for the types included in the batch. 
%
We can obtain a PTAS for the problem \eqref{prob:partition} with type-specific setup times by treating each $J_t$ as a single job and applying a PTAS for the load balancing problem~\cite{HS1987}. Moreover, for the unweighted case, the problem \eqref{prob:partition} is solvable in polynomial time because an optimal solution can be obtained by simply balancing the number of types assigned to each machine. 

\medskip
\noindent\textbf{Library-based setup times}\quad
In cloud computing environments, where each job requires the installation of multiple libraries, and each library takes a fixed amount of time to install. The setup time needed to process a batch of jobs is the total installation time of the required libraries.
Let $L$ be the set of libraries, and let $L_j\subseteq L$ be the set of required libraries for job $j\in J$.
Let $w_\ell$ be the installation time for library $\ell\in L$.
Then, the setup time for $X\subseteq J$ is defined as $c(X)=\sum_{\ell\in \bigcup_{j\in X}L_j}w_\ell$.
Note that if every job $j\in J$ requires a unique library (i.e., $|L_j|=1$), this is a type-specific setup time.
This setup time is a monotone submodular function; more specifically, it is a weighted coverage function.
Hence, we can apply the $\O(\sqrt{n\shortslash\log n})$-approximation algorithm for the uniform submodular load balancing problem~\cite{SvitkinaFleischer2011} to solve the problem \eqref{prob:partition}.

\medskip
\noindent\textbf{TSP-based setup times}\quad
For applications like repair companies, the setup time is defined by the optimal value of a traveling salesman problem (TSP) instance. 
Let $(V,d)$ be a metric space with an origin $o\in V$.
The origin corresponds to the location of the company.
Each job $j\in J$ is associated with a point $v_j\in V$.
Then, the setup time for $X\subseteq J$ is defined as the optimal value of the TSP on $V'=\{v_j\mid j\in X\}\cup\{o\}$.
Note that this class of setup times is a generalization of type-specific setup times.
Indeed, for a star metric $(T\cup\{o\},d)$ where $d(t,o)=w_t/2$ for all $t\in T$ and $d(t,t')=(w_t+w_{t'})/2$ for all distinct $t,t'\in T$, the setup time is $c(X)=\sum_{t\in T:\,X\cap J_t\ne\emptyset}w_t$ for $X\subseteq J$ when $v_j=t$ for each job $j\in J_t$.
For TSP-based setup times, the problem \eqref{prob:partition} can be viewed as the $m$-TSP.
It is known that the $m$-TSP admits a $2.5$-approximation algorithm~\cite{FHK1978,christofides2022worst}.

\subsection{Reduction of the Problem with Release Time}\label{subsec:release}
In this subsection, we show that our assumption that all jobs have a release time of $0$ (i.e., each job is available for processing from time $0$) does not lose generality.
Suppose instead that each job has a release time, and its existence is hidden until its release time.
Let $\ALG$ be a $\rho$-competitive algorithm for a UETS problem in a certain setting without release times.
Then, we demonstrate that $\ALG$ can be transformed into a $(2\rho+1)$-competitive algorithm for the UETS problem with release times of the corresponding setting.
For this end, we use the \emph{IGNORE} strategy, which appeared in the paper of Shmoys et al.~\cite{SWW1995} and is named by Ascheuer et al.~\cite{AKR2000}.

The IGNORE strategy keeps the machines remain idle until a set $S$ of jobs appears.
Then, IGNORE immediately decides a schedule for the jobs in $S$ following the algorithm $\ALG$ and assigns job batches to machines.
We refer to this schedule as a subschedule for $S$.
All jobs that arrive during the process of the subschedule are temporarily ignored until the subschedule for $S$ is completed.
After all the machines complete the jobs in $S$, IGNORE decides a subschedule for all new jobs that arrived during the previous process. If there are no such jobs, all the machines become idle.
The IGNORE strategy repeats this procedure.
\begin{theorem}\label{thm:ignore}
The $\IGNORE$ strategy with a $\rho$-competitive algorithm $\ALG$ is $(2\rho+1)$-competitive for the UETS problem with release time.
\end{theorem}
\begin{proof}
Fix an instance, and let $\OPT$ be the optimal makespan.
For a subset of jobs $J'$, we denote $\ALG(J')$ as the makespan of the algorithm, assuming that jobs in $J'$ are given at time $0$.
Note that $\ALG(J')\le \rho\cdot\OPT$ by the monotonicity of the optimal makespan with respect to the jobs.

If the instance contains no job, then $\IGNORE$ is clearly optimal. Thus, we assume that the instance contains at least one job. 
Let $t^*$ be the last release time of the jobs. Then, $\OPT\ge t^*$ as the last released job can only be processed from $t^*$.

Suppose that the machines are idle at time $t^*$.
Let $R$ be the set of jobs processed in the last subschedule by the IGNORE strategy.
Then, the makespan of IGNORE is
\begin{align*}
t^*+\ALG(R)\le \OPT+\rho\cdot\OPT
\le (2\rho+1)\cdot\OPT.
\end{align*}

On the other hand, suppose that some machines are not idle at time $t^*$.
Then $t^*$ is in the second last subschedule of the algorithm, and the last subschedule starts right after the second-to-last subschedule ends.
Let $R$ and $S$ be the sets of jobs processed in the last and the second-to-last subschedules, respectively.
Then, the makespan of IGNORE is
\begin{align*}
t^*+\ALG(S)+\ALG(R)
\le \OPT+\rho\cdot\OPT+\rho\cdot\OPT=(2\rho+1)\cdot\OPT.
\end{align*}

Therefore, the competitive ratio of IGNORE is at most $2\rho+1$.
\end{proof}

It should be noted that this reduction remains valid even when the competitive ratio $\rho$ of $\ALG$ depends on the number of machines $m$ or the number of jobs $n$, provided it is monotonically nondecreasing with respect to $n$.
From this theorem, the optimum competitive ratio is the same up to a constant factor regardless of the release time.

\section{Single Machine Batch Allocation}\label{subsec:single}
In this section, we analyze the sUETS problem where each batch can be allocated only on a single machine.
Recall that $n$ is the number of jobs and $m$ is the number of machines.

\subsection{Algorithms}\label{subsec:single-alg}
We first observe that the algorithm of processing the batch consisting of all jobs by a single machine without preemption is $m$-competitive, where $m$ is the number of machines.
Note that this algorithm is feasible for any settings of sUETS and mUETS since no cancellation is performed.
\begin{theorem}\label{thm:single-np-m}
The algorithm of processing the batch consisting of all jobs by a single machine is $m$-competitive for the non-preemptive sUETS problem.
\end{theorem}
\begin{proof}
As the makespan of the algorithm is $c(J)+p(J)$ and the optimal makespan is at least $(C(J)+p(J))/m$ by \Cref{lem:opt-lower}, the competitive ratio is at most
\begin{align*}
\frac{c(J)+p(J)}{\max_{i\in[m]}(c(X_i^*)+p(X_i^*))}
&\le \frac{c(J)+p(J)}{\sum_{i\in[m]}(c(X_i^*)+p(X_i^*))/m}
\le \frac{c(J)+p(J)}{(c(J)+p(J))/m}=m,
\end{align*}
where the second inequality is implied by the subadditivity of the setup time $c$.
\end{proof}

Second, we analyze the competitive ratio of the list scheduling algorithm.
\begin{theorem}\label{thm:listscheduling}
The algorithm that greedily assigns a batch consisting of a single unprocessed job to any available single machine is $\O(n/m)$-competitive for the non-preemptive sUETS problem.
\end{theorem}
\begin{proof}
The makespan of the algorithm is at most 
\begin{align*}
\frac{1}{m}\sum_{j\in J}(c(\{j\})+p_j)+\max_{j\in J}(c(\{j\})+p_j)
\le \frac{n\cdot \OPT}{m}+\OPT=\O(n/m)\cdot\OPT,
\end{align*}
where the inequality holds by \Cref{lem:opt-lower}.
Thus, the competitive ratio is at most $\O(n/m)$.
\end{proof}

Next, we provide an $\O\big(\sqrt{n\shortslash m}\big)$-competitive algorithm, which is an improved version of the list scheduling algorithm.
The idea is to reduce the total setup time by grouping jobs.
The algorithm first divides the jobs into $m+\big\lceil\sqrt{mn}\big\rceil$ batches in such a way that this partition minimizes the maximum setup time, under the condition that each batch contains at most $\big\lceil\sqrt{n\shortslash m}\big\rceil$ jobs.
Then, it processes the batches according to the list scheduling algorithm, which greedily allocates an unprocessed batch to any available machine.
The formal description of this algorithm is summarized in~\Cref{alg:single-np}.

\begin{algorithm}[ht]
  Let $k=m+\big\lceil\sqrt{mn}\big\rceil$\;
  Compute a $k$-partition $(X_1,\dots,X_k)$ of the jobs $J$ such that $|X_i|\le \big\lceil\sqrt{n\shortslash m}\big\rceil$ for all $i\in[k]$, and among such partitions, minimize $\max_{i\in[k]} c(X_i)$\;\label{line:size_limited_partition}
  \For{$i\ot 1,2,\dots, k$}{
    Wait until at least one machine is available\;
    Assign batch $X_i$ to an available machine\;
  }
  \caption{$\O(\sqrt{n\shortslash m})$-competitive algorithm for the non-preemptive sUETS problem}\label{alg:single-np}
\end{algorithm}

\begin{theorem}\label{thm:single-np-nm}
\Cref{alg:single-np} is $\O\big(\sqrt{n\shortslash m}\big)$-competitive for the non-preemptive sUETS problem.
\end{theorem}
\begin{proof}
Let $(X_1^*,\dots,X_m^*)$ be the optimal schedule and let $\OPT$ denote the optimal makespan $\max_{i\in[m]}(c(X_i^*)+p(X_i^*))$.
Define $k=m+\big\lceil\sqrt{mn}\big\rceil$, and let $(X_1,\dots,X_k)$ be a $k$-partition of the jobs $J$ that minimizes  $\max_{i\in[k]} c(X_i)$ subject to $|X_i|\le\big\lceil\sqrt{n\shortslash m}\big\rceil$ for all $i\in[k]$.

We first observe that $\max_{i\in[k]}c(X_i)\le \OPT$.
This is because by dividing each $X_i^*$ with $i\in[m]$ into $\big\lceil |X_i^*|/\sqrt{n\shortslash m}\big\rceil$ sub-batches of almost equal size, we can obtain a partition $(Y_1,\dots,Y_k)$ that refines $(X_1^*,\dots,X_m^*)$.
Here, we have $|Y_{i'}|\le \max_{i\in[m]}\Big\lceil|X_i^*|\big/\big\lceil|X_i^*|/\sqrt{n\shortslash m}\big\rceil \Big\rceil\le \big\lceil\sqrt{n\shortslash m}\big\rceil$ for all $i'\in[k]$.
Thus, it holds that 
$\max_{i\in[k]}c(X_i)
\le \max_{i\in[k]}c(Y_i)
\le \max_{i\in[m]}c(X^*_i)
\le \OPT$.

Then, the makespan of the algorithm is at most 
\begin{align*}
\MoveEqLeft
\frac{1}{m}\sum_{i\in[k]}(c(X_i)+p(X_i))+\max_{i\in[k]}(c(X_i)+p(X_i))\\
&\le \frac{k\cdot \OPT+p(J)}{m}+\max_{i\in[k]}|X_i|\cdot\max_{j\in J}(c(\{j\})+p_j)\\
&\le \Big(k/m+1+\big\lceil\sqrt{n\shortslash m}\big\rceil\Big)\cdot\OPT
=\O\big(\sqrt{n\shortslash m}\big)\cdot\OPT,
\end{align*}
where the first inequality follows from the subadditivity of the setup time and the second inequality from \Cref{lem:opt-lower}.
Hence, \Cref{alg:single-np} is $\O\big(\sqrt{n\shortslash m}\big)$-competitive.
\end{proof}

By combining \Cref{thm:single-np-m,thm:single-np-nm}, we obtain the following corollary.
\begin{corollary}\label{cor:single-np-n}
There exists an $\O(n^{1/3})$-competitive algorithm for the non-preemptive sUETS problem.
\end{corollary}
\begin{proof}
If $m\le n^{1/3}$, then the algorithm that processes the batch $J$ by a single machine is $m=\O(n^{1/3})$-competitive by \Cref{thm:single-np-m}.
On the other hand, if $m> n^{1/3}$, the competitive ratio of \Cref{alg:single-np} is $\O\big(\sqrt{n\shortslash m}\big)=\O(n^{1/3})$ by \Cref{thm:single-np-nm}.
Thus, an $\O(n^{1/3})$-competitive algorithm exists for either case.
\end{proof}

\begin{remark}\label{rem:sUETS}
Obtaining the partition at line~\ref{line:size_limited_partition} in Algorithm~\ref{alg:single-np} is computationally hard.
Thus, we mention the competitive ratio when we can utilize an $\alpha$-approximation algorithm for the problem~\eqref{prob:partition}.
Let $(X'_1,\dots,X'_m)$ be an $m$-partition of the jobs $J$ that is obtained by the approximation algorithm.
Let $k'=m+\big\lceil\sqrt{mn\shortslash\alpha}\big\rceil$.
Then, we can easily construct a $k'$-partition $(Y'_1,\dots,Y'_{k'})$ that is a refinement of $(X'_1,\dots,X'_m)$ and satisfies $|Y'_i|\le \big\lceil\sqrt{\alpha\cdot n\shortslash m}\big\rceil$ for all $i\in [k']$.
By using $k'$ instead of $k$ and employing this partition at line~\ref{line:size_limited_partition}, 
\Cref{alg:single-np} is $\O\big(\alpha+\sqrt{\alpha\cdot n\shortslash m}\big)$-competitive for the non-preemptive sUETS problem.
By combining this with \Cref{thm:listscheduling}, we can obtain an $\O(\min\{\sqrt{\alpha\cdot n\shortslash m},\,n/m\})$-competitive algorithm.
Moreover, by combining this with \Cref{thm:single-np-m}, we can also obtain an $\O(\min\{(\alpha n)^{1/3},\sqrt{n}\})$-competitive algorithm.
\end{remark}

In the rest of this subsection, we provide an $\O(\log n/\log\log n)$-competitive algorithm for the preemptive sUETS problem. 
For a given instance of the sUETS problem, let $q$ be an integer such that $q^q\ge n>(q-1)^{q-1}$.
Note that $q=\Theta(\log n/\log\log n)$.
The execution of our algorithm consists of $q+1$ phases.
Without loss of generality, we may assume that $m\ge 2q$, since otherwise the competitive ratio of the algorithm in \Cref{thm:single-np-m} is $m=\O(q)=\O(\log n/\log\log n)$.

The intuition of our algorithm is as follows.
Our algorithm repeatedly completes a $(1-1/q)$-fraction of the jobs in each phase. Consequently, all the jobs can be processed in $\O(q)$ iterations. This increases the cost of the setup time by a factor of $\O(q)$.
Additionally, we ensure that at least $m/q$ machines work at any point of the algorithm.
This guarantees that the cost in terms of execution time will only increase by a factor of at most $q$.
By these properties, the competitive ratio of the algorithm 
is $\O(q)=\O(\log n/\log\log n)$.

Let us explain our algorithm more precisely.
Initially, it computes an $m$-partition $(X_1,\dots,X_m)$ of the jobs $J$ that minimizes the maximum setup time $\max_{i\in[m]}c(X_i)$.
In the first phase, $X_i$ is allocated to machine $i$ for each $i\in[m]$.
Then, each machine processes the assigned batch until either the machine completes the batch or the number of uncompleted machines becomes less than or equal to $\lfloor m/q\rfloor$.
At that time, the machines still processing will preempt their batches. 
In the $k$th phase ($k=2,3,\dots,q$), the algorithm computes an $m$-partition $(Y_1^{(k)},\dots,Y_m^{(k)})$ of the remaining jobs such that, for any $i\in[m]$, $|Y_i^{(k)}|\le n/q^{k-1}$ and $Y_i^{(k)}\subseteq X_j$ for some $j\in[m]$.
We can always obtain such a partition by dividing each batch left in the previous phase into $q$ sub-batches.
Thereafter, each machine processes the assigned batch until either the batch is completed or the number of uncompleted machines becomes less than or equal to $\lfloor m/q\rfloor$.
At the end of $q$th phase, the number of remaining jobs is at most $(n/q^{q-1})\cdot\lfloor m/q\rfloor\le m\cdot(n/q^q)\le m$.
In the $(q+1)$st phase, the algorithm assigns up to one remaining job to each machine.
A formal description of the algorithm is provided in \Cref{alg:single-p}.

\begin{algorithm}[ht]
  Let $q$ be an integer such that $q^q\ge n>(q-1)^{q-1}$\;
  \lIf{$m<2q$}{Assign batch $J$ to one machine and exit}
  Compute an $m$-partition $(X_1,\dots,X_m)$ of the jobs $J$ that minimizes $\max_{i\in[m]} c(X_i)$\;
  Let $R^{(0)}\ot J$ be the set of uncompleted jobs\;
  \For(\tcp*[h]{phase $k$}){$k\ot 1,2,\dots, q$}{
    Compute an $m$-partition of $(Y_1^{(k)},\dots,Y_m^{(k)})$ of the remaining jobs $R^{(k-1)}$ such that, for each $i\in[m]$, $|Y_i^{(k)}|\le n/q^{k-1}$ and $Y_i^{(k)}\subseteq X_j$ for some $j\in[m]$\;\label{line:partition}
    Assign batch $Y_i^{(k)}$ to machine $i$ for each $i\in[m]$\;
    Each machine processes the assigned batch until the batch is completed or the number of uncompleted machines becomes less than or equal to $\lfloor m/q\rfloor$\;\label{line:wait}
    Let $R^{(k)}$ be the set of uncompleted jobs\;
  }
  Assign up to one job $R^{(q)}$ per machine\tcp*[l]{phase $q+1$}\label{line:last}
  \caption{$\O\big(\frac{\log n}{\log\log n}\big)$-competitive algorithm for the preemptive sUETS problem}\label{alg:single-p}
\end{algorithm}

\begin{theorem}\label{thm:single-p-n}
\Cref{alg:single-p} is $\O\big(\frac{\log n}{\log\log n}\big)$-competitive for the preemptive sUETS problem.
\end{theorem}
\begin{proof}
Let $q$ be an integer such that $q^q\ge n>(q-1)^{q-1}$.
If $m<2q$, then the competitive ratio is $m=\O(q)=\O(\log n/\log\log n)$ by \Cref{thm:single-np-m}.
Thus, we assume $m\ge 2q$.

Let $\OPT$ be the optimal makespan and let $\pmax = \max_{j\in J} p_j$.
In addition, let $\tau_k$ be the time length of phase $k\in[q+1]$.
By the definition of the algorithm, the jobs in $R^{(k-1)}\setminus R^{(k)}$ are completed in phase $k\in[q]$.
Thus, the total execution time of jobs in phase $k\in[q]$ is at most $p(R^{(k-1)}\setminus R^{(k)})+\pmax\cdot\lfloor m/q\rfloor$ because at most $\lfloor m/q\rfloor$ jobs are partially executed.
Taking into account the setup time of at most $\max_{i\in[m]}c(X_i)$ and the time for preemption of at most $\max_{i\in[m]}c(X_i)+p_{\max}$, the length of phase $k\in[q]$ is
\begin{align*}
\tau_k
&\le \max_{i\in[m]}c(X_i)+\frac{p(R^{(k-1)}\setminus R^{(k)})+\pmax\cdot\lfloor m/q\rfloor}{\lfloor m/q\rfloor}+\left(\max_{i\in[m]}c(X_i)+p_{\max}\right)\\
&\le \frac{q}{m-q}\cdot p(R^{(k-1)}\setminus R^{(k)})+4\OPT
\le \frac{2q}{m}\cdot p(R^{(k-1)}\setminus R^{(k)})+4\OPT,
\end{align*}
where the second inequality is by \Cref{lem:opt-lower} and the third is by $m\ge 2q$.
In addition, by \Cref{lem:opt-lower}, the time length $\tau_{q+1}$ of phase $q+1$ is
\begin{align*}
\tau_{q+1}\le \max_{j\in J}(c(\{j\})+p_j)\le \OPT.
\end{align*}
Thus, the makespan of \Cref{alg:single-p} is
\begin{align*}
\sum_{k=1}^{q+1} \tau_k
&\le (4q+1)\OPT+\frac{2q}{m}\cdot \sum_{k=1}^q p(R^{(k-1)}\setminus R^{(k)})\\
&\le (4q+1)\OPT+\frac{2q}{m}\cdot p(J)
\le (4q+1)\OPT+2q\cdot \OPT=\O(q)\cdot\OPT.
\end{align*}
Hence, the competitive ratio of \Cref{alg:single-p} is $\O(q)=\O(\log n/\log\log n)$.
\end{proof}

\begin{remark}
Similar to~\Cref{rem:sUETS}, suppose that we take an $m$-partition $(X_1,\dots,X_m)$ of the jobs $J$ that is an $\alpha$-approximation of the problem \eqref{prob:partition} at line~\ref{line:partition} in \Cref{alg:single-p}.
Then, the competitive ratio of \Cref{alg:single-p} becomes $\O\big(\alpha\cdot\frac{\log n}{\log\log n}\big)$-competitive for the preemptive sUETS problem.
Moreover, by setting $q$ to be the integer that satisfies $q^q\ge n^\alpha>(q-1)^{q-1}$, we can slightly improve the competitive ratio to $\O\big(\frac{\log n^\alpha}{\log\log n^\alpha}\big)$.
\end{remark}

\subsection{Lower bounds}\label{subsec:single-lower}
In this subsection, we show that the algorithms provided in \Cref{subsec:single-alg} are asymptotically optimal with respect to both the number of machines and the number of jobs.

We first provide a lower bound for the non-preemptive setting.
\begin{theorem}\label{thm:single-np-hard}
The competitive ratio of the non-preemptive sUETS problem is at least $\Omega(m)$ and $\Omega(n^{1/3})$ even for constant setup times.
\end{theorem}
\begin{proof}
Let $n=m^3$ and $c(X)=1$ for any non-empty batch $X\subseteq J$ (i.e., constant setup times).
We will set the execution time $1$ for at most $m$ jobs (referred to as \emph{heavy jobs}) and $0$ for the other jobs (referred to as \emph{light jobs}) depending on the behavior of a given online algorithm.
It is not difficult to see that the optimal makespan is at most $2$.

We fix an online algorithm.
Let us consider its behavior when every job is light. 
Suppose that every batch assigned to machines consists of at most $m$ jobs.
Then, the algorithm assigns at least $n/m=m^2$ batches to the machines in total. Consequently, some machines must process at least $m$ batches, which requires $m$ setups.
Thus, the makespan is at least $m$ in this case, and we have done.

Conversely, suppose that the algorithm allocates a batch consisting of more than $m$ jobs in the instance where every job is light. 
In another instance where $m$ jobs in the first such batch are set to be heavy, the algorithm's behavior is the same until the batch is assigned.
Thus, the algorithm assigns all the heavy jobs to one machine, and the makespan is at least $m$.
Therefore, in both cases, the competitive ratio is at least $\Omega(m)=\Omega(n^{1/3})$.
\end{proof}

Next, we provide a lower bound for the preemptive setting.
\begin{theorem}\label{thm:single-p-hard}
The competitive ratio of the preemptive sUETS problem is at least 
$\Omega(m)$ and $\Omega(\log n/\log\log n)$ even for constant setup times.
\end{theorem}
\begin{proof}
Let $n=(m+1)^m$. 
For every non-empty batch $X\subseteq J$, let $c(X)=1$ (i.e., constant setup times).
We set the execution time $1$ for $m$ jobs (referred to as 
 \emph{heavy jobs}) and $0$ for the remaining jobs (referred to as \emph{light jobs}).
The selection of heavy jobs depends on the online algorithm.
In addition, when an algorithm processes a batch containing heavy jobs, we treat that heavy jobs are executed first and light jobs are executed later.
It is not difficult to see that the optimal makespan is $2$.

We fix an online algorithm.
Let $X^{(1)}_i$ be the set of jobs assigned to machine $i\in[m]$ just before time $1$. Moreover, let $X^{(1)}_0$ be the set of jobs that are not assigned to any machine just before time $1$.
Let $i^{(1)}\in \argmax_{i\in\{0,1,\dots,m\}} |X^{(1)}_i|$.
Note that $|X^{(1)}_{i^{(1)}}|\ge (m+1)^{m-1}$ by $n=(m+1)^m$.
We set that all the heavy jobs are contained in $X^{(1)}_{i^{(1)}}$.
Then, just before time $2$, the number of remaining jobs is at least $(m+1)^{m-1}$, and the number of executed heavy jobs is at most $1$.

Similarly, for each $t=2,3,\dots,m$, we define 
$X^{(t)}_i$ as the set of jobs assigned to machine $i\in[m]$, and 
$X^{(t)}_0$ as the set of jobs that are remaining but not assigned to any machine, just before time $t$.
Let $i^{(t)}\in \argmax_{i\in\{0,1,\dots,m\}} |X^{(t)}_i|$.
Then, we have $|X^{(t)}_{i^{(t)}}|\ge (m+1)^{m-t}$.
We set that the remaining heavy jobs are contained in $X^{(t)}_{i^{(t)}}$.
Then, just before time $t+1$, the number of remaining jobs is at least $(m+1)^{m-t}$, and the number of executed heavy jobs is at most $t$.

For this instance, the makespan of the schedule obtained by the online algorithm is at least $m$.
Therefore, the competitive ratio is at least $m/2=\Omega(m)=\Omega(\log n/\log\log n)$.
\end{proof}

\section{Multiple Machines Batch Allocation}
In this section, we explore the mUETS problem, where a batch can be allocated to multiple machines.

\subsection{Algorithms}
We first provide an $\O(\sqrt{m})$-competitive algorithm for the non-preemptive mUETS problem.
The algorithm computes a $\lfloor\sqrt{m}\rfloor$-partition of $J$ that minimizes maximum setup time. Then, it allocates each batch to $\lfloor\sqrt{m}\rfloor$ machines.
The algorithm is formally described in \Cref{alg:multiple-np}.

\begin{algorithm}[t]
  Let $k=\lfloor \sqrt{m}\rfloor$\;
  Compute a $k$-partition $(X_1,\dots,X_{k})$ of the jobs $J$ that minimize $\max_{i\in[k]} c(X_i)$\;
  Assign batch $X_i$ to $\lfloor m/k\rfloor$ machines for each $i\in[k]$\;
  \caption{$\O(\sqrt{m})$-competitive algorithm for the non-preemptive mUETS problem}\label{alg:multiple-np}
\end{algorithm}

\begin{theorem}\label{thm:multiple-np}
\Cref{alg:multiple-np} is $\O(\sqrt{m})$-competitive for the non-preemptive mUETS problem.
\end{theorem}
\begin{proof}
Let $(X_1^*,\dots,X_m^*)$ be an optimal schedule and let $\OPT$ be the optimal makespan.
In addition, let $k=\lfloor\sqrt{m}\rfloor$ and $(X_1,\dots,X_k)$ be a $k$-partition of $J$ that minimizes $\max_{i\in[k]} c(X_i)$.
We analyze the algorithm that allocates each batch $X_i$ to $\lfloor m/k\rfloor$ machines.
The makespan of the algorithm is at most $\max_{i\in[k]} \left(c(X_i)+\frac{p(X_i)}{\lfloor m/k\rfloor}+\max_{j\in X_i}p_j\right)$.
By the choice of $(X_1,\dots,X_k)$ and the subadditivity of the setup time $c$, we have
\begin{align}
\begin{split} 
\max_{i\in[k]} c(X_i)
&\le \max_{i\in[k]} c\left(\bigcup_{j=1}^{\lceil m/k\rceil}X^*_{\lceil m/k\rceil\cdot (i-1)+j}\right)\\
&\le \max_{i\in[k]} \sum_{j=1}^{\lceil m/k\rceil}c(X^*_{\lceil m/k\rceil\cdot (i-1)+j})
\le \left\lceil\frac{m}{k}\right\rceil\cdot\OPT, 
\end{split}
\label{ineq:multinp}
\end{align}
where we denote $X^*_i = \emptyset$ for $i>k$.
Therefore, by \Cref{lem:opt-lower} and \eqref{ineq:multinp}, the makespan is at most
\begin{align*}
\max_{i\in[k]} \left(c(X_i)+\frac{p(X_i)}{\lfloor m/k\rfloor}+\max_{j\in X_i}p_j\right)
&\le \left(\left\lceil\frac{m}{k}\right\rceil+\frac{m}{\lfloor m/k\rfloor}+1\right)\cdot\OPT
=\O(\sqrt{m})\cdot\OPT,
\end{align*}
which means that the competitive ratio of the algorithm is $\O(\sqrt{m})$.
\end{proof}

By combining \Cref{thm:multiple-np,thm:single-np-nm}, we can obtain the following corollary.
\begin{corollary}\label{cor:multiple-np-n}
There is an $\O(n^{1/4})$-competitive algorithm for the non-preemptive mUETS problem.
\end{corollary}
\begin{proof}
If $m\le \sqrt{n}$, then the algorithm in \Cref{thm:multiple-np} is $\O(\sqrt{m})=\O(n^{1/4})$-competitive.
On the other hand, if $m>\sqrt{n}$, the competitive ratio of \Cref{thm:single-np-nm} is $\O(\sqrt{n\shortslash m})=\O(n^{1/4})$.
Thus, in either case, there is an $\O(n^{1/4})$-competitive algorithm.
\end{proof}

\begin{remark}
To account for the computational issue, suppose that we only have an $m$-partition $(X'_1,\dots,X'_m)$ of the jobs $J$ that is an $\alpha$-approximation for the problem \eqref{prob:partition}.
Let $k'=\lfloor\sqrt{\alpha m}\rfloor$.
Then, the schedule that assigns batch $\bigcup_{j=1}^{\lceil m/k'\rceil}X'_{\lceil m/k'\rceil\cdot(i-1)+j}$ to $\lfloor m/k'\rfloor$ machines for each $i\in[k']$ is $\O(\sqrt{\alpha m})$-competitive for the non-preemptive mUETS problem.
By combining this with \Cref{rem:sUETS}, we can obtain an $\O(\min\{\sqrt{\alpha}\cdot n^{1/4},\,\sqrt{n}\})$-competitive algorithm for the non-preemptive mUETS problem.
\end{remark}

\begin{algorithm}[t]
  Let $q$ be an integer such that $q^q\ge m>(q-1)^{q-1}$\;
  Compute an $m$-partition $X_1,\dots,X_m$ of the jobs $J$ that minimizes $\max_{i\in[m]} c(X_i)$\;\label{line:multiple-p-partition}
  Set $k^*\ot \lfloor\log_q m\rfloor+1$ and let $I^{(0)}\ot [m]$ be the set of unprocessed batch indices\;
  \For(\tcp*[h]{phase $k$}){$k\ot 1,2,\dots, k^*$}{
    Assign batch $X_i$ to $q^{k-1}$ machines for each $i\in I^{(k-1)}$\;
    Continue processing the assignment until the number of uncompleted batches becomes $\lfloor m/q^k\rfloor$. Once this condition is met, preempt all remaining uncompleted batches\;
    Define $I^{(k)}$ to be the set of uncompleted batch indices\;
  }
  \caption{$\O\big(\frac{\log m}{\log\log m}\big)$-competitive algorithm for the preemptive mUETS problem}\label{alg:multiple-p}
\end{algorithm}

For the preemptive mUETS problem, \Cref{alg:single-p} is also $\O(\log n/\log\log n)$-competitive, and this is asymptotically best possible as we will see in \Cref{thm:multiple-p-hard}.
Thus, even when we are allowed to assign a job batch to more than one machine, we cannot improve the competitive ratio with respect to the number $n$ of jobs.
In contrast, we show that the competitive ratio with respect to the number $m$ of machines can be exponentially improved by allowing batches to be assigned to multiple machines.

Let $q$ be an integer such that $q^q> m\ge (q-1)^{q-1}$ and let $k^*$ be $\lfloor\log_q m\rfloor+1$.
Here, $k^*\le q$ and $q^{k^*-1}\le m<q^{k^*}$.
We have $q\ge 2$ from the assumption that $m\ge 2$.

Our algorithm for the preemptive mUETS is similar to \Cref{alg:single-p} for the preemptive sUETS.
However, instead of splitting the uncompleted batch into several smaller batches of similar size, the algorithm increases the number of machines assigned to that batch.
This ensures that the competitive ratio of the algorithm does not depend on the number of jobs $n$.

Our algorithm computes an $m$-partition $(X_1,\dots,X_m)$ of the jobs $J$ that minimizes the maximum setup time $\max_{i\in[m]}c(X_i)$.
In the first phase, each batch $X_i$ is allocated to one machine.
Then, each batch is processed until it is completed, or the number of uncompleted batches becomes less than or equal to $\lfloor m/q\rfloor$.
In the $k$th phase ($k=2,3,\dots,k^*$), each uncompleted batch $X_i$ is assigned to $q^{k-1}$ machines, and it is processed until it is completed, or the number of uncompleted batches becomes $\lfloor m/q^k\rfloor$.
If the number of uncompleted batches becomes smaller than $\lfloor m/q^k\rfloor$ because multiple batches are completed simultaneously, then break the tie arbitrarily and defer the rest batches to the next phase.
Note that this allocation is feasible because $\lfloor m/q^{k-1}\rfloor\cdot q^{k-1}\le m$.
At the end of the $k^*$th phase, all the batches are completed since $\lfloor m/q^{k^*}\rfloor=0$ by $k^*=\lfloor \log_q m\rfloor+1>\log_q m$.
Our algorithm is formally described in \Cref{alg:multiple-p}.
It should be noted that this algorithm is based on a similar idea to that of \cite[Algorithm 1]{goko2022online}.

\begin{theorem}\label{thm:multiple-p}
\Cref{alg:multiple-p} is $\O\big(\frac{\log m}{\log\log m}\big)$-competitive for the preemptive mUETS problem.
\end{theorem}
\begin{proof}
Let $\OPT$ be the optimal makespan and let $\pmax = \max_{j\in J} p_j$.
For each $k\in[k^*]$, let $S^{(k)}=I^{(k-1)}\setminus I^{(k)}$ be the set of indices of batches that have been completed in the $k$th phase.
By the definition of the algorithm, we have
\begin{align*}
|S^{(k)}|
&=\left\lfloor \frac{m}{q^{k-1}}\right\rfloor-\left\lfloor \frac{m}{q^{k}}\right\rfloor
\ge \left\lfloor \frac{m}{q^{k-1}}\right\rfloor-\frac{1}{q}\cdot\left\lfloor \frac{m}{q^{k-1}}\right\rfloor
= \left(1-\frac{1}{q}\right)\left\lfloor\frac{m}{q^{k-1}}\right\rfloor
\ge \frac{1}{2}\cdot\left\lfloor\frac{m}{q^{k-1}}\right\rfloor.
% &=\lfloor m/q^{k-1}\rfloor-\lfloor m/q^{k}\rfloor\\
% &\ge \lfloor m/q^{k-1}\rfloor-(1/q)\cdot\lfloor m/q^{k-1}\rfloor
% = (1-1/q)\lfloor m/q^{k-1}\rfloor
% \ge \frac{1}{2}\cdot\lfloor m/q^{k-1}\rfloor.
\end{align*}

Let $\tau_k$ be the time length of phase $k\in[k^*]$.
We first bound it for $k\in[k^*-1]$.
As every batch in $S^{(k+1)}$ is not finished in phase $k$ and the preemption takes time at most $\max_{i\in S^{(k)}} (c(X_i)+\max_{j\in X_i}p_j) \leq 2\OPT$, we have
\begin{align*}
\tau_k
&\le \min_{i\in S^{(k+1)}} \left(c(X_i)+\frac{p(X_i)}{q^{k-1}}+\pmax\right)+2\OPT\\
&\le \frac{1}{|S^{(k+1)}|}\cdot\sum_{i\in S_{k+1}} \left(c(X_i)+\frac{p(X_i)}{q^{k-1}}+\pmax\right)+2\OPT\\
&\le \sum_{i\in S^{(k+1)}}\frac{1}{|S^{(k+1)}|}\cdot \frac{p(X_i)}{q^{k-1}}+4\OPT
\le \frac{m/q^{k}}{\frac{1}{2}\cdot\lfloor m/q^{k}\rfloor}\cdot\sum_{i\in S^{(k+1)}}q\cdot\frac{p(X_i)}{m}+4\OPT.
\end{align*}
Note that $x/\lfloor x\rfloor\le 2$ if $x\ge 1$ and $m/q^k\ge 1$ for $k\in [k^*-1]$.
Thus, we have 
\begin{align}
\tau_k\le \sum_{i\in S^{(k+1)}}4q\cdot\frac{p(X_i)}{m}+4\OPT \label{eq:bound1}
\end{align}
for every $k\in[k^*-1]$.

Next, we bound the time length $\tau_{k^*}$ of phase $k^*$. 
As the batch $X_i$ for each $i\in S^{(k^*)}$ is processed by $q^{k^*-1}$ machines, we have
\begin{align}
\tau_{k^*}
&\le \max_{i\in S^{(k^*)}} \left(c(X_i)+\frac{p(X_i)}{q^{k^*-1}}+\pmax\right)\notag\\
&\le \frac{p(J)}{q^{k^*-1}}+2\OPT
< q\cdot \frac{p(J)}{m}+2\OPT
\le (q+2)\OPT. \label{eq:bound2}
\end{align}

Hence, by \eqref{eq:bound1} and \eqref{eq:bound2}, 
the makespan of the algorithm is 
\begin{align*}
\sum_{k=1}^{k^*}\tau_k
&\le \sum_{k=1}^{k^*-1}\bigg(\sum_{i\in S^{(k+1)}}4q\cdot\frac{p(X_i)}{m}+4\OPT\bigg)+(q+2)\OPT\\
&\le 4q\cdot\frac{p(J)}{m}+4(q-1)\OPT+(q+2)\OPT
\le (9q-2)\OPT=\O(q)\cdot\OPT.
\end{align*}
Therefore, the competitive ratio of \Cref{alg:multiple-p} is $\O(q)=\O(\log m/\log\log m)$.
\end{proof}

\begin{remark}
Suppose that we take an $m$-partition $(X_1,\dots,X_m)$ of the jobs $J$ that is an $\alpha$-approximation for the problem \eqref{prob:partition} at line~\ref{line:multiple-p-partition} in \Cref{alg:single-p}.
Then, the competitive ratio of \Cref{alg:multiple-p} becomes $\O\big(\alpha\cdot\frac{\log m}{\log\log m}\big)$-competitive for the preemptive mUETS problem.
Moreover, by setting $q$ to be the integer that satisfies $q^q\ge m^\alpha>(q-1)^{q-1}$, we can slightly improve the competitive ratio to $\O\big(\frac{\log m^\alpha}{\log\log m^\alpha}\big)$.
\end{remark}

\subsection{Lower Bounds}
In this subsection, we show the asymptotic optimality of our algorithms.
We first provide lower bounds for the non-preemptive case.
\begin{theorem}\label{thm:multiple-np-hard}
The competitive ratio of the non-preemptive mUETS is at least $\Omega(\sqrt{m})$ and $\Omega(n^{1/4})$ even for unweighted type-specific setup times.
\end{theorem}
\begin{proof}
For each type $i\in[m]$, we prepare a group of $m$ jobs $J_i=\{(i-1)\cdot m+j\mid j\in[m]\}$.
The set of jobs is $J=\bigcup_{i=1}^m J_i$, and hence $n=m^2$.
We employ the unweighted type-specific setup times, i.e., $c(X)=|\{i\in[m]\mid X\cap J_i\ne\emptyset\}|$ for each batch $X\subseteq J$.
We will set the execution time $1$ for at most $m$ jobs (referred to as \emph{heavy jobs}) and $0$ for the remaining jobs (referred to as \emph{light jobs}).
The selection of heavy jobs depends on the online algorithm.
The optimal makespan is at most $3$ because the heavy jobs can be processed in time $2$, and the light jobs can be processed in time $1$.

We fix an online algorithm. 
Suppose that every job is light.
Let $X_1,\dots,X_r$ be the series of non-empty batches that the algorithm assigns to machines.
% the algorithm completes the jobs by processing batches of $X_1,\dots,X_r$. 
Note that $(X_1,\dots,X_r)$ is a partition of $J$.
In addition, suppose that the algorithm assigns batch $X_\ell$ to a set of $s_\ell$ machines for each $\ell\in [r]$.
Then, by evaluating the total setup time, the makespan is at least
\begin{align*}
\frac{1}{m}\cdot\sum_{\ell\in[r]} c(X_\ell)\cdot s_\ell
=\frac{1}{m}\cdot\sum_{\ell\in[r]} |\{i\in[m] \mid X_\ell\cap J_i\ne\emptyset\}|\cdot s_\ell.
\end{align*}

Let $L=\{\ell\in[r]\mid s_\ell\ge \sqrt{m}\}$.
We consider three cases based on the behavior of the algorithm.
First, if $\sum_{\ell\in L} |\{i\in[m] \mid X_\ell\cap J_i\ne\emptyset\}|\ge m/2$, then the makespan is at least
\begin{align*}
\frac{1}{m}\cdot\sum_{\ell\in L} |\{i\in[m] \mid X_\ell\cap J_i\ne\emptyset\}|\cdot \sqrt{m}
\ge \frac{1}{m}\cdot \frac{m}{2}\cdot\sqrt{m}
=\frac{\sqrt{m}}{2}=\Omega(\sqrt{m}).
\end{align*}

Next, if $|X_\ell|>\sqrt{m}\cdot s_\ell$ for some $\ell\in [r]\setminus L$, consider an instance where $\sqrt{m}\cdot s_\ell~(\le m)$ jobs in one of such $X_\ell$ are set to be heavy.
Then, in this case, the makespan of the algorithm for this instance is at least $|X_\ell|/s_\ell\ge \sqrt{m}$.

Finally, assume that $\sum_{\ell\in L} |\{i\in[m] \mid X_\ell\cap J_i\ne\emptyset\}|< m/2$ and $|X_\ell|\le \sqrt{m}\cdot s_\ell$ for every $\ell\in [r]\setminus L$.
We analyze the makespan of the instance where every job is light.
Let $I=\{i\in [m]\mid X_\ell\cap J_i=\emptyset~(\forall \ell\in L)\}$ be the index set of groups that are disjoint from $X_\ell$ for all $\ell \in L$.
% By $\sum_{\ell\in L} |\{i\in[m] \mid X_\ell\cap J_i\ne\emptyset\}|< m/2$, the number of groups not containing a job to be processed in a batch included in $L$ is at least $m/2$.
Since $\sum_{\ell\in L} |\{i\in[m] \mid X_\ell\cap J_i\ne\emptyset\}|< m/2$, 
% Thus, letting $I=\{i\in [m]\mid X_\ell\cap J_i=\emptyset~(\forall \ell\in L)\}$, 
we have $|I|\ge m/2$.
Consequently, we have
\begin{align*}
\sum_{\ell\in [r]\setminus L} s_\ell
&\ge \frac{1}{\sqrt{m}}\sum_{\ell\in [r]\setminus L}|X_\ell|
=\frac{1}{\sqrt{m}}\cdot|\{j\in J\mid j\not\in X_\ell~(\forall \ell\in L)\}|\\
&\ge \frac{1}{\sqrt{m}}\left|\bigcup_{i\in I}J_i\right|
\ge \frac{1}{\sqrt{m}}\cdot\frac{m}{2}\cdot m 
= \frac{m^{3/2}}{2}.
\end{align*}
This implies that the makespan is at least
\begin{align*}
\frac{1}{m}\cdot\sum_{\ell\in[r]} |\{j\in[m] \mid X_\ell\cap J_j\ne\emptyset\}|\cdot s_\ell
&\ge \frac{1}{m}\cdot\sum_{\ell\in [r]\setminus L} |\{j\in[m] \mid X_\ell\cap J_j\ne\emptyset\}|\cdot s_\ell\\
&= \frac{1}{m}\cdot\sum_{\ell\in [r]\setminus L} s_\ell
\ge \frac{1}{m}\cdot\frac{m^{3/2}}{2} = \frac{\sqrt{m}}{2} = \Omega(\sqrt{m}).
\end{align*}

Therefore, in all cases, the makespan of the algorithm is at least $\Omega(\sqrt{m})$ while the optimal makespan is at most $3$. 
This implies that the competitive ratio is at least $\Omega(\sqrt{m})$.
In addition, as $n=m^2$, the competitive ratio is also at least $\Omega(n^{1/4})$.
\end{proof}

Next, we provide lower bounds for the preemptive case.
\begin{theorem}\label{thm:multiple-p-hard}
The competitive ratio of the preemptive mUETS problem is at least $\Omega(\log m/\log\log m)$ and $\Omega(\log n/\log\log n)$ even for unweighted type-specific setup times.
\end{theorem}
\begin{proof}
Let $q$ be the positive integer such that $q^q\le m< (q+1)^{q+1}$.
For each type $i\in[q^q]$, we prepare a group of jobs $J_i=\{(i-1)\cdot m^2+j\mid j\in[m^2]\}$.
The set of jobs $J$ is defined as $\bigcup_{i=1}^{q^q}J_i$. Hence, $n=|J|=q^q\cdot m^2~(\le m^3)$.
We employ the unweighted type-specific setup times, i.e., $c(X)=|\{i\in[q^q]\mid X\cap J_i\ne\emptyset\}|$ for each batch $X\subseteq J$.
The execution time of each job will be set to be $1$ or $0$, depending on the online algorithm.
We refer to jobs with the execution time of $1$ and $0$ as \emph{heavy} and \emph{light}, respectively.

When an algorithm assigns a batch $X\subseteq J$ to $k$ machines, we assume that the batch is processed as follows.
Let $i_1,\dots,i_\ell$ be the types of jobs in the batch $X$, i.e., $\{i\in[q^q]\mid J_i\cap X\ne\emptyset\}=\{i_1,\dots,i_\ell\}$.
Initially, each of the $k$ machines prepares to execute jobs in $X\cap J_{i_1}$ with a setup time of $1$.
If a machine is available and there is a remaining (uncompleted and not in progress) job in $X\cap J_{i_1}$, it starts executing that job. If there is no such job, the machine prepares to execute jobs in $X\cap J_{i_2}$ with a setup time of $1$.
Each machine continues this process, executing jobs in the order of groups $i_1,\dots,i_\ell$.
Note that such an assumption does not make the problem more difficult than the assumption that every machine incurs a setup time $c(X)$ at the beginning.

We fix an online algorithm.
We partition the groups $[q^q]$ into $q+1$ sets $S_1,S_2,\dots,S_{q+1}$ based on the behavior of the algorithm, as we will describe below.
For $t\in[q+1]$, each job $j\in J_i$ with $i\in S_t$ is set to heavy if it begins to be executed strictly before time $t$.

For each time $t$ and $i\in[q^q]$, let $\kappa_i(t)$ be the number of machines that are executing or setting up to execute jobs in $J_i$ just before time $t$.
Define $S_1\subseteq [q^q]$ to be the set of groups with the $q^q-q^{q-1}$ largest values of $\kappa_i(1)$.
Hence, $|S_1|=q^q-q^{q-1}$ and $\kappa_i(1)\ge \kappa_{i'}(1)$ for any $i\in S_1$ and $i'\not\in S_1$.
Similarly, for each $t=2,3,\dots,q$, define $S_t$ sequentially to be the set of $i\in [q^q]\setminus\bigcup_{\tau=1}^{t-1} S_{\tau}$ such that $|S_t|=q^{q+1-t}-q^{q-t}$ and $\kappa_i(t)\ge \kappa_{i'}(t)$ for any $i \in S_t$ and $i'\in [q^q]\setminus\bigcup_{\tau=1}^{t} S_{\tau}$.
Finally, define $S_{q+1}$ to be the set of remaining indices, i.e., $S_{q+1}=[q^q]\setminus\bigcup_{\tau=1}^{q}S_{\tau}$.
We remark that $S_t$'s are disjoint, and $S_{q+1}$ is a singleton because $\sum_{\tau=1}^q |S_{\tau}|=\sum_{\tau=1}^q (q^{q+1-\tau}-q^{q-\tau})=q^q-1$.
Let $S_{q+1}=\{i^*\}$.
Since some jobs in $J_{i^*}$ must be executed at time $q+1$ or later, the makespan of the online algorithm is at least $q+1$. 

For each $t\in [q]$, let $J^{(t)}$ be the set of heavy jobs that were started to be executed in the time interval $[t,t+1)$.
We evaluate the cardinality of $J^{(t)}$ for each $t$.
Recall that it takes one unit of time to execute a heavy job. 
Hence, each machine can start to execute at most one job in $J^{(t)}$.
In addition, if a machine executes a job in $J_i$ in the time interval $[t,t+1)$, the machine must be executing or setting up to execute jobs in $J_i$ at just before time $t$.
As the heavy jobs executed in the time interval $[t,t+1)$ belong to a group in $\bigcup_{\tau=t+1}^{q+1}S_{\tau}=[q^q]\setminus\bigcup_{\tau=1}^tS_{\tau}$, and $\kappa_i(t) \leq \kappa_{i'}(t)$ for $i\in S_t$ and $i'\in[q^q]\setminus\bigcup_{\tau=1}^t S_\tau$, we have
\begin{align*}
|J^{(t)}|
\le {\textstyle\sum_{i\in \bigcup_{\tau=t+1}^{q+1}S_{\tau}}\kappa_i(t)}
\le m\cdot \frac{\big|\bigcup_{\tau=t+1}^{q+1}S_{\tau}\big|} {\big|\bigcup_{\tau=t}^{q+1}S_{\tau}\big|}
= m\cdot \frac{q^{q-t}}{q^{q+1-t}} = \frac{m}{q}.
\end{align*}
Here, recall that $\kappa_i(t) \leq \kappa_{i'}(t)$ for all $i\in \bigcup_{\tau=t+1}^{q+1} S_{\tau}$ and $i'\in S_t$.

As the machines start to process heavy jobs at a time in $[1,q+1)=\bigcup_{\tau=1}^q[\tau,\tau+1)$, the total number of heavy jobs is $\sum_{t\in[q]}|J^{(t)}|\le (m/q)\cdot q=m$.
Hence, the optimal makespan is at most $3$ because the heavy jobs can be processed in time $2$, and the light jobs can be processed in time $1$.
Therefore, the competitive ratio is at least $(q+1)/3=\Omega(\log m/\log\log m)=\Omega(\log n/\log\log n)$.
\end{proof}

\section{Conclusion and Discussion}\label{sec:discuss}
In this paper, we introduced the UETS problem and studied it in terms of competitive analysis. We obtained tight bounds of the competitive ratio in each setting.
In the following, we discuss the application of our results to variant settings.

When a batch $X\subseteq J$ is assigned to $k$ machines, we have assumed that every machine incurs a setup time $c(X)$.  
However, for instance, in the production system example, unnecessary reassignment of attachments can be skipped.
Additionally, if a machine processes multiple batches, the setup time incurred from the second or subsequent batch may be reduced. 
Our algorithmic results remain valid even for these cases as long as an optimal schedule assigns only one batch to each machine separately.
This is because, for algorithms, a decrease in setup time only contributes to decreasing the makespan.
%However, our algorithmic results and lower bounds remain valid even if this assumption does not hold as long as an optimal schedule does not assign a batch to multiple machines. 
%For algorithms, a decrease in setup time only contributes to decreasing makespan, thus the results remain valid. 
% For lower bounds, we only use unweighted type-specific setup times in the proofs. 
% Such setup times can be seen as an abstraction for the scheduling in a production system. 
% For these assignments, the setup time cannot be reduced, as, for example, skipping unnecessary reassignment of attachments is not possible.
% Thus, our lower bounds hold even without the assumption.
%As our competitive analysis provides guarantees for the worst-case scenario, it is applicable even in such a case. 
%In such cases, it is sufficient for the algorithm to process only batches that consist of jobs of a single type. 

In the preemptive setting, we assumed that a job completed at the time of preemption does not need to be executed again.
We can also consider a setting in which, if a batch process is preempted, the entire batch must be restarted from scratch. This means that even completed jobs in the batch must be executed again.
\Cref{alg:multiple-p} also works in this setting, and it is $\O(\log m/\log\log m)$-competitive because its analysis does not use the information of which jobs are completed.
However, \Cref{alg:single-p} may not achieve the competitive ratio of  $\O(\log n/\log\log n)$ as the proof does not work for this setting.
Nevertheless, if we can obtain information on the execution time of completed jobs, we can reestablish an $\O(\log n/\log\log n)$-competitive algorithm by processing the jobs that are completed once but need to be reprocessed between phases.

\section{Acknowledgments}
This work was partially supported by the joint project of Kyoto University and Toyota Motor Corporation, titled ``Advanced Mathematical Science for Mobility Society'', 
JST ERATO Grant Number JPMJER2301,
JST PRESTO Grant Number JPMJPR2122, %YK
and 
JSPS KAKENHI Grant Numbers 
JP17K12646, %HS
JP19K22841, %KM
JP20H00609, %KM
JP20H05967, %KM
JP20K19739, %YK
JP21K17708, %HS
JP21H03397, and %HS
JP25K00137. %YK, HS

\bibliographystyle{abbrvnat}
\bibliography{online}

\end{document}